\documentclass[12pt]{article}
\usepackage{amsmath,amsfonts,amsthm,amssymb, graphicx,color,hyperref}
\usepackage{amsthm,amsfonts,amsmath, amscd}
\usepackage{amssymb,amsmath, dsfont} \usepackage{graphicx}
\swapnumbers
\newtheorem{thm}{THEOREM}[section] \newtheorem{lm}[thm]{LEMMA}

 \theoremstyle{definition}

\newcommand{\tr}{{\rm Tr}} 
\renewcommand{\|}{{\Vert}} 
 \numberwithin{equation}{section}
\newcommand{\R}{{\mathord{\mathbb R}}}

\newcommand{\Z}{{\mathord{\mathbb Z}}}

\begin{document}

\def\tr{{\rm Tr}}

\title{Approach to the
    steady state in kinetic models with thermal reservoirs at different temperatures}

\author{\vspace{5pt} E. A. Carlen$^1$,  R. Esposito$^3$, J. L. Lebowitz$^{1,2}$,\\ R. Marra$^4$ and
  C. Mouhot$^{5}$ \\
  \vspace{5pt}\small{$1.$ Department of Mathematics, $2.$ Department
    of Physics,}\\
  [-6pt]\small{Rutgers University, 110 Frelinghuysen Road, Piscataway NJ 08854-8019 USA}\\
  \small{$3.$ M\&MOCS, Univ. dell'Aquila, Cisterna di Latina, (LT) 04012 Italy}\\
  \small{$4.$ Dipart. di Fisica and Unit\`a INFN, Universit\`a di Roma Tor Vergata, 00133 Roma, Italy}\\
\vspace{5pt}\small{$5.$  DPMMS, University of Cambridge, Wilberforce Road, Cambridge CB3 0WA, UK}\\
}

\date{}

\maketitle

\medskip

\let\thefootnote\relax\footnote{ \copyright \, 2017 by the
authors. This paper may be reproduced, in its entirety, for
non-commercial purposes.}

\begin{abstract} We continue the  investigation of  kinetic models of a system in contact via stochastic interactions
with several spatially homogeneous thermal reservoirs at different temperatures. Considering models different from those investigated in \cite{CLM}, we explicitly compute the unique spatially uniform non-equilibrium steady state (NESS) and prove that it is approached exponentially fast from any uniform initial state. This leaves open the question of whether there exist NESS that are not spatially uniform. Making a further simplification of our models, we then  prove non-existence of such NESS and exponential approach to the unique spatially uniform NESS (with a computably  boundable  rate). The method of proof  relies on refined Doeblin estimates and other probabilisitic techniques, and is quite different form the analysis in \cite{CLM} that was based on contraction mapping methods. 
\end{abstract}

\section{Introduction} \label{intro}

We investigate the time evolution and {\em non-equilibrium steady states} (NESS)
of a gas, described on the mesoscopic scale   by a one particle phase space probability distribution $f(x,v,t)$ in contact with heat reservoirs at
different temperatures. The models considered here and the methods of analysis are different from this considered in \cite{CLM} by some of the same authors.

Here $x\in \Lambda$, where $\Lambda$ is a $d$-dimensional torus of side length $L$,
and $v \in \R^d$. Starting with some
initial state $f(x,v,0)$, $f$  changes in time according to an autonomous equation of the  general form
\begin{equation}\label{eq1}
\frac{\partial}{\partial t} f(x,v,t)    +{\rm div}_x(vf(x,v,t) = Q[f](x,v,t) + \sum_{j=1}^k L_j f(x,v,t)\  .           
\end{equation}

$Q$ accounts for the effects of binary molecular collisions, as in the Boltzmann collision kernel, or other interactions between the particles. The $L_j$
account for the effects of interactions with the thermal reservoirs. 

In the absence of the reservoirs, (\ref{eq1}) reduces to 
\begin{equation}\label{eq1a}
\frac{\partial}{\partial t} f(x,v,t)    +{\rm div}(vf(x,v,t) = Q[f](x,v,t)\  .           
\end{equation}
which would be the Boltzmann equation were $Q$ a Boltzmann collision kernel. In all of the models we discuss, $Q$ will be such that solutions of
(\ref{eq1a}) conserve the total energy, and their only steady states are the global (spatially uniform) Maxwellian  densities 
\begin{equation}\label{MTdef}
M_{T,u}(x,v) =|\Lambda|^{-1}(2\pi T)^{-d/2} e^{-|v-u|^2/2T}\ ,
\end{equation}
where $|\Lambda|$ denotes the volume of the torus $\Lambda$, $T>0$ and $u\in \R^d$. Under certain assumptions on the solutions for the actual Boltzmann equation, and more generally for the simplifications that we study here, it is known \cite{AEP,DV,GMM,Vil} that the solution $f(x,v,t)$ to (\ref{eq1a}) converges  to 
$M_T(v)$ at a rate that in the most recent of these references is shown to be exponential,  The particular values of $T$ and $u$ are fixed by the initial data $f_0(x,v)$ and  momentum and energy  conservation:
$$u = \int_{\R^d} vf_0(x,v){\rm d}v \qquad{\rm and}\qquad  T = \frac{1}{d}\int_{\R^d} |v-u|^2f_0(x,v){\rm d}v\ .$$

Inclusion of several reservoirs at different temperatures complicates matters considerably. In (\ref{eq1}),
each  $L_j$  is the adjoint of the generator of the jump process modeling interactions with the $j$th thermal reservoir at temperature $T_j$:  For $f\in L^1(\Lambda\times\R^d)$,
\begin{equation}\label{Ljdef}
L_jf(x,v) = \eta_j [m_{T_j}(v)\rho_f(x) - f(x,v)]  
\end{equation}
where 
\begin{equation}\label{mTdef}
m_T(v) =(2\pi T)^{-d/2} e^{-|v|^2/2T} \quad{\rm and}\quad \rho_f(x) = \int_{\R^d} f(x,v){\rm d}v\ .
\end{equation}
Note that $m_T$ is a probability density on $\R^d$, unlike $M_T$ which is a probability density on $\Lambda \times \R^d$. 
The process described by $L_j$ is the jump process in which  jump times arrive in Poisson stream, and then there occurs a jump in phase space from $(x,v)$ to $(x,w)$ -- keeping the position fixed -- where $w$ is chosen from the distribution $m_{T_j}(w) {\rm d}w$, independent of $v$.

If $Q$ were absent in (\ref{eq1}), this equation would reduce to 
\begin{equation}\label{eq1b}
\frac{\partial}{\partial t} f(x,v,t)    +{\rm div}(vf(x,v,t) =  \sum_{j=1}^k L_j f(x,v,t)\  .           
\end{equation}
This is a relatively tractable equation, and is a special case of a type of equation considered in \cite{BL,LB}.  In this case, the only steady state is 
\begin{equation}\label{combined}
\frac{1}{|\Lambda|} \left(\sum_{j=1}^k \eta_j\right)^{-1}\sum_{j=1}^k \eta_j m_{T_j}(v)\ .
\end{equation}
That is, in contrast with (\ref{eq1a}), the long time behavior is independent of the initial data since energy and 
momentum are not conserved, but are regulated by the reservoirs. For this equation it is also known that 
solutions approach the steady state exponentially fast \cite{BL}.  Such non-Maxwellian steady states are non-equilibrium steady states; they are not Maxwellian equilibrium states.  

When there is only a single reservoir at temperature $T$, $M_T$ is the unique steady state for (\ref{eq1}). 
This may be proved by a simple modification of the  entropy argument that identifies the steady states of (\ref{eq1a}); see, e.g., the book  of Cercignani \cite{Cer}.
The {\em relative entropy} of $f$ with respect to  the Maxwellian steady state $f_\star := |\Lambda|^{-1}m_T(v)$ is the quantity
\begin{equation}\label{entprod}
H(f|f_\star)  := \int_{\Lambda \times \R^d} f \log f {\rm d}x {\rm d}v  -  \int_{\Lambda \times \R^d} f \log f_\star {\rm d}x {\rm d}v\ .
\end{equation}
The rate of change of $D(f||f_\star)$ along any solution $f(x,v,t)$ to (\ref{eq1}) with a singe reservoir is the quantity
${\displaystyle {\mathcal D} (f||f_\star) := \frac{\rm d}{{\rm d}t}H(f|f_\star)  }$, and 
\begin{equation}\label{relent}
{\mathcal D} (f||f_\star)  =  \int_{\Lambda \times \R^d}  \log f (Q[f]+L[f]) {\rm d}x {\rm d}v  -  
\int_{\Lambda \times \R^d}  \log f_\star  (Q[f]+L[f]) {\rm d}x {\rm d}v\ .
\end{equation}

Since
$\log f_\star$ is a linear combination of $|v|^2$ and the constant functions $1$, and since mass and energy are conserved by $Q$,
${\displaystyle \int_{\Lambda \times \R^d}  \log f Q[f] {\rm d}x {\rm d}v = 0}$. Moreover, whenever $f$ satisfies
\begin{equation}\label{enbal}
\int_{\Lambda \times \R^d} |v|^2 f(x,v) {\rm d}x {\rm d}v = dT
\end{equation}
${\displaystyle \int_{\Lambda \times \R^d}  \log f L[f] {\rm d}x {\rm d}v = 0}$.
It is easy to see that for {\em any} steady state $f$  of (\ref{eq1}) with a single reservoir at temperature $T$, (\ref{enbal}) must be satisfied.   Then  for any such steady state $f$,
$$
{\mathcal D} (f||f_\star)  =  \int_{\Lambda \times \R^d}  \log f (Q[f]+L[f]) {\rm d}x {\rm d}v\ .
$$ 
By Boltzmann's $H$-Theorem, 
$\int_{\Lambda \times \R^d}  \log f  Q[f]{\rm d}x {\rm d}v \leq 0$. A simple calculation shows that 
$\int_{\Lambda \times \R^d}  \log f  L[f]{\rm d}x {\rm d}v \leq 0$
with equality if and only if $f$ has the form $\rho(x)m_T(v)$. 

Since the dissipation must vanish in any steady state, every steady state must have the form $f(x,v) = 
\rho(x)m_T(v)$. But then $0 = {\rm div}_x(\rho(x)m_T(v)) = v\cdot \nabla\rho(x)m_T(v)$. Hence $\rho$ is constant, and then
$f= f_\star$.  Moreover, this entropy argument can be extended, using the relative entropy as a Lyapunov function, to bound the rate of convergence to the unique steady state. 

However, when there is more than one temperature, $|\Lambda|^{-1}m_T(v)$  will not be a steady state  
for any temperature $T$, and then  the entropy argument sketched above, which depended on the fact that 
$\log(|\Lambda|^{-1}m_T(v))$ is a 
linear combination of $|v|^2$ and the constant function $1$, is no longer applicable.  
The steady states to (\ref{eq1}), which do exist, but are not Maxwellian, are non-equilibrium steady states (NESS).

In short, when there are reservoirs with more than one temperature, Boltzmann's 
$H$-Theorem no longer is of use for determining the
steady sates, or estimating the rate of approach to them, and in this case, much less is known about the steady states and the approach to them 
unless the temperatures are extremely close.  

The situation in which the temperatures of the reservoirs are close can be treated perturbatively, but this is far from simple.
In physical situations, the different reservoirs would occupy different 
spatial regions, a natural case being
that in which  Maxwell boundary conditions with different temperatures are imposed on different parts 
of the boundary of the domain $\Lambda$. Recently   Esposito et.~al. \cite{EGKM} proved  that when
$Q[f]$ is a Boltzmann collision term and the temperatures along the boundary  are all close to some fixed temperature  $T$, then there is an NESS close to $M_{T,0}$ and when the initial data
$f(x,v,0)$ is close
to $M_{T,0}$, then $f(x,v,t)$ will indeed approach  $M_{T,0}$. 
Even in this case, one does not know if there are not other NESS, and 
the methods in \cite{EGKM} do not extend to the case where the initial state is not close to $M_{T,0}$.

Here we treat simplified models in which the reservoirs act not at the boundaries, but throughout the volume of 
$\Lambda$ in a spatially uniform manner.  We consider this treatment as a first step toward a non-perturbative 
analysis of systems with reservoirs at substantially different temperatures.  The model may, however, have some direct physical relevance. One might consider the reservoirs to interact with the particles through collisions  with thermal baths of photons that permeate $\Lambda$. 
Since our reservoirs act not only at the boundary of $\Lambda$, but uniformly throughout 
$\Lambda$, we conjecture that there are no spatially non-uniform NESS for these models. However, 
at present, we can only prove this for relatively simple models that we now precisely describe.

\subsection{Thermostatted  kinetic   equations}

Let $\Lambda$ be a torus in $\R^d$ with volume $|\Lambda|=L^d$. (The rate of approach to equilibrium will depend on $L$; see Section 3.1.)
For any time dependent one particle distribution $f(x,v,t)$ on the phase space $\Lambda\times\R^d$,
define the {\em hydrodynamic moments}
\begin{eqnarray}\label{moments}
 \rho(x,t) &=& \int_{\R^d}f(x,v,t){\rm d}v \nonumber\\
 u(x,t) &=& \frac{1}{\rho(x,t)} 
\int_{\R^d}vf(x,v,t){\rm d}v \nonumber\\
T(x,t) &=& \frac{1}{d\rho(x,t)} \int_{\R^d}|v- u(x,t)|^2 f(x,v,t){\rm d}v \ .
\end{eqnarray}
(Note that $u(x,t)$ and $T(x,t)$  are only defined when $\rho(x,t)\neq 0$.)
Define the instantaneous (at time $t$)  {\em local Maxwellian corresponding to} $f$, denoted by $M_f(x,v,t)$,
as
\begin{equation}\label{locmax}
M_f(x,v,t) = \rho(x,t) (2\pi T(x,t))^{-d/2} \exp( - |v- u(x,t)|^2/2T(x,t))\ ,
\end{equation}
with the natural convention that $M_f(x,v,t)=0$ when $\rho(x,t)= 0$.

We  consider two forms of $Q$. The first is a kinetic (self consistent) Fokker-Planck form, 
and the second is of BGK form.  

The {\em kinetic Fokker-Planck equation} (KFP) is the self-consistent equation
\begin{equation}\label{kfp}
\frac{\partial}{\partial t} f(x,v,t)   +{\rm div}_x(vf(x,v,t))  =
 T(x,t) {\rm div}_v \left( M_f(x,v,t) )\nabla_v \frac{f(x,v,t)}{M_f(x,v,t)}\right)\ ,
\end{equation}
where a diffusion constant has been absorbed into the time scale for convenience.
The right hand side of (\ref{kfp}) can be written as
\begin{equation}\label{kfp2}
\mathcal{G}_f f(x,v,t) := T(x,t)\Delta_v f(x,v,t) + {\rm div}_v\left((v- u(x,t))f(x,v,t)\right)\ .
\end{equation}

The BGK form of the equation is
\begin{equation}\label{bgk}
\frac{\partial}{\partial t} f(x,v,t)   +{\rm div}_x(vf(x,v,t))  =
\alpha [M_f(x,v,t)  -f(x,v,t)]\ 
\end{equation}
for a constant $\alpha>0$.

Both of these  equations conserve total energy, momentum and mass. Without loss of generality, we restrict our
attention to initial data (for these equations and others to be considered) $f_0$ such that
\begin{equation}\label{zeromom}
\int_{\Lambda\times \R^d} vf(x,v,t){\rm d}x{\rm d}v = 0\ .
\end{equation}
Under this condition, the only spatially homogeneous steady states of  (\ref{kfp}) and (\ref{bgk}) are the {\em global Maxwellian} phase space probability densities of the form
\begin{equation}\label{MTdef}
M_T(x,v) = \frac{1}{|\Lambda|}(2\pi T)^{-d/2} e^{-|v|^2/2T}\ .
\end{equation}

The {\em thermostatted Fokker-Planck equation} is 
\begin{equation}\label{tkfp}
\frac{\partial}{\partial t} f(x,v,t)   +{\rm div}_x(vf(x,v,t) ) =
 \mathcal{G}_f f(x,v,t)) +  \sum_{j=1}^k L_j f(x,v,t)  
\end{equation}
where $\mathcal{G}_f f(x,v,t)$ is given by (\ref{kfp2}). 
The thermostatted BGK equation is 
\begin{equation}\label{tbgk}
\frac{\partial}{\partial t} f(x,v,t)   +{\rm div}_x(vf(x,v,t))  =
\alpha[M_f(x,v,t) - f(x,v,t)] +  \sum_{j=1}^k L_j f(x,v,t)  \ .
\end{equation}

Let $f_*(x,v)$ be any steady state solution of (\ref{tkfp}). 
Multiplying by  a smooth function $\phi(x)$ and integrating over phase space, we obtain
\begin{equation}\label{ste2}
\int_\Lambda \nabla \phi(x)\cdot u(x) \rho(x){\rm d}x = 0\ .
\end{equation}
In two or three dimensions, this says that the divergence of $\rho u$ is zero, and in one dimension it says that
$\rho u$ is constant, and then under the condition (\ref{zeromom}),  $u(x) = 0$ for all $x$. 

Next, in one dimension, multiplying by $\phi(x)v$ and integrating over phase space
\begin{equation}\label{ste21}
\int_\Lambda \nabla \phi(x)\rho(x) T(x){\rm d}x = 0\ ,
\end{equation}
so that $\rho(x)T(x)$, the pressure,  is constant. But even in one dimension, this does not imply that $\rho(x)$ and $T(x)$ are individually constant.  The same remarks apply in the BGK case as well.

There is one simple result that can be proved about convergence to a steady state for solutions of (\ref{tkfp}) and (\ref{tbgk}):
More precisely,   for a phase space density $f$ define
\begin{equation}\label{Tdef}
T_f  = \frac{1}{d}\int_\Lambda \int_{\R^d} |v|^2 f(x,v){\rm d}v{\rm d}x\ .
\end{equation}
and let $M_{T_f}$ denote the global Maxwellian with temperature $T_f$; i.e.,  the probability density given by (\ref{MTdef}) with $T =T_f$.

\begin{lm}\label{templem2}For any solution $f(x,v,t)$ of (\ref{tkfp}) and(\ref{tbgk})
\begin{equation}\label{temprate2}
\frac{{\rm d}}{{\rm d}t} T_{f}(t) = \sum_{j=1}^d \eta_j \left( T_j - T_f(t)\right)\ .
\end{equation}
Therefore, if we define the quantities $\eta$ and $T_\infty$ by
\begin{equation}\label{Tinf}
\eta := \sum_{j=1}^k \eta_j \qquad{\rm and}\qquad T_\infty = \frac{1}{ \eta}\sum_{j=1}^k \eta_j T_j\  ,
\end{equation}
we have that
\begin{equation}\label{Tinf20}
 T_f(t) = T_\infty + e^{-t\eta}(T_{f}(0) - T_\infty) \  .
\end{equation}
\end{lm}

\begin{proof}Note that
${\displaystyle \int_V \int_{\R^d} |v|^2 {\rm div}_x(vf(x,v,t) ){\rm d}v {\rm d}x = 0}$
so that the term representing the effects of spatial inhomogeneity drops out of (\ref{tkfp}). Also, since energy is conserved globally by the Fokker-Planck term. 
$$\int_V \int_{\R^d} |v|^2 G_{T_f(t)}[f](x,v,t){\rm d}v {\rm d}x = 0\ .$$
The rest follows from the definition of the reservoir terms. The analysis for (\ref{tbgk}) is essentially the same. Solving the equation
(\ref{temprate2}) yields (\ref{Tinf}) and (\ref{Tinf20}).
\end{proof}

In particular, Lemma~\ref{templem2} says that for any solution $f(x,v,t)$ of (\ref{tkfp}) or (\ref{tbgk}) ${\displaystyle \lim_{t\to\infty}T_f(t) = T_\infty}$ and the convergence is exponentially fast.

In Section \ref{sect2}, 
we use Lemma~\ref{templem2}, among other devices, to determine the explicit forms of all   {\em spatially homogeneous} steady states for  both thermostatted equations. 

To go beyond the determination of the spatially uniform steady states, and to prove exponential convergence to them from general initial data, or even that there are no steady states that are not spatially uniform, we  simplify our model. We modify $Q$ so that energy 
is conserved globally but not locally: We  replace $M_f(x,v,t)$ in (\ref{kfp}) by
$$\rho(x,t) M_{T_f(t)}(v)\ .$$  
Defining the operator $G_T$ by
\begin{equation}\label{FPdef}
G_T[g](v) =  T{\rm div}_v \left(M_T \nabla \frac{f}{M_T}\right) = T\Delta g(v) + 
{\rm div}[vg(v)] \ ,
\end{equation}
we may write the resulting equation as
\begin{equation}\label{eq30}
\frac{\partial}{\partial t} f(x,v,t)    +{\rm div}_x(vf(x,v,t)) = G_{T_f(t)}[f](x,v,t)  + \sum_{j=1}^k L_j f(x,v,t)\ ,        
\end{equation}

Likewise, in the BGK case,  we replace $M_f(x,v,t)$ in the gain term of $Q$ by $\rho(x,t) M_{T_f(t)}(v)$. The resulting equation is
\begin{multline}\label{tbgk22}
\frac{\partial}{\partial t} f(x,v,t)   +{\rm div}_x(vf(x,v,t))  =\\
\alpha[\rho_f(x,t)M_{T_f(t)}(v,t) - f(x,v,t)] +  \sum_{j=1}^k L_j f(x,v,t)  
\end{multline}

Equations (\ref{eq30}) and (\ref{tbgk22}) are still non-linear, but only in a superficial way. 
This means that we may regard $T_f(t)$  in equations (\ref{eq30}) or (\ref{tbgk22}) as {\it a-priori}  known, and then these equations become the ({\em linear}) Kolmogorov forward equations of Markov processes with time-dependent generators.   
This allows us to apply probabilisitic methods to the problem of uniqueness of steady states, and to the problem of 
proving exponentially fast relaxation to the steady states. 

Our main results are the following:

\begin{thm}\label{allunif}  Let $f_*(x,v)$ be an NESS for (\ref{eq30}) or (\ref{tbgk22}). Then $f_*(x,v)$ does not depend on $x$. That is, every NESS for (\ref{eq30}) and (\ref{tbgk22}) is spatially uniform. 
\end{thm}

\begin{thm}\label{approach} Let $f$ be a solution of  either (\ref{eq30}) or (\ref{tbgk22}) with $T_f(0) < \infty$. Let $f_\infty$ be the 
unique stationary state, which for solutions of (\ref{eq30}) is given by Theorem~\ref{ssf}, and for solutions of (\ref{tbgk22})
is given by Lemma~\ref{bgksss}.  Then  there are finite positive and explicitly computable  constants $C$ and $c$ such that
$$\int_{\Lambda\times \R^d}|f(x,v,t) - f_\infty(v)|{\rm d}x{\rm d}v \leq Ce^{-ct}\ .$$

\end{thm}

Theorem~\ref{approach} is proved in Section \ref{sect4}.
We emphasize that the method provides an explicitly computable bound on the rate of convergence that is likely to be useful elsewhere. 

The proof of Theorem~\ref{allunif} is much simpler than that of Theorem~\ref{approach}. By an ergodicity argument, all we need to do is to find {\em one} spatially homogeneous steady state. Ergodicity then implies that there are no others. In Section 2 we find the spatially homogeneous steady states for (\ref{eq30}) or (\ref{tbgk22}). Granted that these exist, we may prove Theorem~\ref{allunif} as follows:

\begin{proof}[Proof of Theorem~\ref{allunif}] By the convergence of $T_f(t)$ to $T_\infty$, any steady state solution of (\ref{eq30}) must also be a steady state solution of
\begin{equation}\label{eq4}
\frac{\partial}{\partial t} f(x,v,t)  +{\rm div}_x(vf(x,v,t)    = G_{T_\infty}[f](x,v,t)  + \sum_{j=1}^k L_j f(x,v,t)\ .            
\end{equation}
In Theorem~\ref{ssf} below we find the explicit form of the unique {\em spatially homogeneous} stationary state of this equation.  It remains to show that
there are no other stationary states of any kind. However, this is an immediate consequence of the fact that 
(\ref{eq4}) is the forward equation of an ergodic process. Lower bounds for the transition kernel $p_t((x,v),(y,w))$ that display this ergodicity will be given below in Section 3.1.  The same reasoning aplies  to the BGK model.
\end{proof}

\section{Explicit form of the spatially homogeneous NESS}\label{sect2}

If the initial data  $f_0$ for either (\ref{tkfp}) or (\ref{tbgk})  is spatially homogenous; 
i.e., translation invariant on $\Lambda$, then the solution $f(x,v,t)$ of (\ref{kfp}) or (\ref{tbgk})
will be spatially homogeneous for all $t$, 
so that $f(x,v,t) = |\Lambda|^{-1}g(v,t)$ for a time dependent probability density $g(v,t)$. In this case,
\begin{equation}
M_f(x,v,t)  = M_{T_f(t)} \quad{\rm with}\quad T_g(t) = d^{-1}\int_{\R^d}|v|^2g(v,t){\rm d}v\ .
\end{equation} 

It is simplest to write down the stationary states for (\ref{tbgk}). Under the assumptions of stationarity and spatial homogeneity,
(\ref{tbgk}) reduces to
$$
[M_{T_\infty} -g] + \sum_{j-1}^k \eta_j[ M_{t_j} - g] = 0
$$
which immediately yields
\begin{lm}\label{bgksss} The unique normalized spatially homogeneous steady state solution of (\ref{tbgk}), $g_\infty(v)$, is given by
\begin{equation}\label{bgkness}
g_\infty(v) = |\Lambda|^{-1}\left(1+ \sum_{j=1}^k \eta_j\right)^{-1}\left(M_{T_\infty} +  \sum_{j=1}^k \eta_j M_{T_j}\right)\ 
\end{equation}
where $T_\infty$ is given by (\ref{Tinf20}).
\end{lm}

An explicit formula can also be given in the kinetic Fokker-Planck case, but its derivation is not so immediate. 
Under the assumptions of stationarity and spatial homogeneity,  (\ref{tkfp})  reduces to 
\begin{equation}\label{eqred}
\frac{\partial}{\partial t} g(v,t)    =
 G_{T_g(t)} g(v,t) +  \sum_{j=1}^k L_j g(v,t)  
\end{equation} 
where for $T>0$, the operator $G_T$ is given by (\ref{FPdef}). 

By Lemma~\ref{templem2},   any steady state solution $g_\infty$ of (\ref{eqred}) must satisfy
\begin{equation}\label{eqred2}
 G_{T_\infty } g_\infty(v) +  \sum_{j=1}^k L_j g_\infty(v)   = 0
\end{equation} 
where $T_\infty$ is given by (\ref{Tinf}).

The equation
\begin{equation}\label{eq4xx}
\frac{\partial}{\partial t} g(v,t)    = G_{T_\infty}[g](v,t)  + \sum_{j=1}^k L_j g(v,t)\ ,          
\end{equation}
 is the forward Kolmogorov equation of a stochastic process $v_t$ in $\R^d$ that has the following description:
At random times $t$, arriving in a Poisson stream with rate $\eta$, there are interactions with  the $k$ reservoirs.
When an interaction occurs, an index $j\in \{1,\dots,k\}$ is chosen with probability $\eta_j/\eta$, and then
$v_t$ jumps to a new point chosen from the distribution $M_{T_j}$.  Between interactions with the reservoir, the particle diffuses, governed by the SDE
\begin{equation}\label{SDE}
{\rm d}v_t = -v_t{\rm d}t  + \sqrt{2T_\infty }{\rm d}w_t
\end{equation}
where $w_t$ is a standard Brownian motion. 

Consider a large time $t_1$, so that with very high probability, there has been at least one collision. 
Almost surely, there are at most finitely many collisions. Let $\hat t_1$ be the {\em last} collision before time $t_1$. 
Then at time $(\hat t_1)_+$, supposing  the $j$-th reservoir is selected for the interaction, the conditional distribution is given by $M_{T_j}$. If one starts a solution of the SDE  (\ref{SDE}) with initial distribution $M_{T_j}$, at each later time
the distribution is $M_T$ for some $T$ in between $T_j$ and $T_\infty$.   This heuristic suggest that 
the invariant measure of the process governed by (\ref{eq4xx}) is a convex combination of Gaussians $M_T$ 
with $T \in [T_1,T_k]$, where we have assumed, without loss of generality that $T_1 < T_2 < \cdots < T_k$. 
More precisely, we expect an invariant density $f_\infty$ (\ref{eq4xx})  of the form
\begin{equation}\label{eq4f}
f_\infty(v) = |\Lambda|^{-1}\int_{[T_1,T_k]} M_T(v){\rm d}\nu(T)\             
\end{equation}

Of course, if $k>2$ and $T_j = T_\infty$ for some $T_j\in (T_1,T_2)$, one can expect $\nu$ to have a point mass
at $T_j = T_\infty$, since if at the last interaction with the reservoirs the particle has jumped to a point distributed according to $M_{T_\infty}$, the diffusion does not change this distribution. Otherwise, we expect $\nu$ to be absolutely continuous,
so that for some probability density  $w(T)$ on $[T_1,T_k]$, 
\begin{equation}\label{eq4fv}
{\rm d}\nu(T) = w(T){\rm d}T\ .        
\end{equation}

One has to be careful about conditioning a Markov process on a random time that depends on future events (such as the time of the {\em last} interaction with the reservoirs before time $t_1$). The heuristic argument  put forward can be made rigorous in several ways, using the fact that the Poisson stream of interaction times is independent of the diffusion process. 
However, since we need the explicit form of the probability measure $\nu$ in what follows, it is simplest to treat (\ref{eq4f}) as an ansatz,
and to derive the form of $\nu$. The next theorem gives the explicit form of the unique steady state $g_\infty$ of (\ref{eq4xx}) for $k=2$. 

\begin{thm}[Steady state formula]\label{ssf}  Suppose $k=2$ with $T_1 < T_2$, and $\eta_1,\eta_2> 0$. 
Then there is a unique steady state solution $g_\infty$ of (\ref{eq4xx}) which is given by (\ref{eq4f})  and (\ref{eq4fv}) where
\begin{equation}\label{eq4fss}
w(T)  = \begin{cases} {\displaystyle  \frac{\eta_1}{2(T_\infty -T_1)^{\eta/2}}(T_\infty - T)^{\eta/2-1}} & T\in [T_1,T_\infty)\\
 {\displaystyle  \frac{\eta_2}{2(T_2 -T_\infty )^{\eta/2}} (T-T_\infty)^{\eta/2-1} }& T\in (T_\infty,T_2]\end{cases}\ .      
\end{equation}
\end{thm}

\begin{proof} Let 
\begin{equation}\label{sss0}
g(v) := \int_{T_1}^{T_2} w(T)M_T{\rm d}T
\end{equation}
 where $w(T)$ is a probability density on $[T_1,T_2]$.
If $g$ is to be a steady state solution of (\ref{eq4xx}), then we must have
\begin{equation}\label{sss1}
G_{T_\infty} g + \eta_1 M_{T_1}  + \eta_2 M_{T_2} - \eta g = 0\  .
\end{equation}
Note that 
$$G_{T_\infty}M_T = G_{T}M_T + (T_\infty -T)\Delta M_T = (T_\infty -T)\Delta M_T \ .$$
Then since
$$\frac{\partial}{\partial T} M_T(v) = 2\Delta M_T(v)\ ,$$
\begin{eqnarray}\label{sss2}
G_{T_\infty} g  &=& 2\int_{T_1}^{T_2} (T_\infty - T)w(T) \frac{\partial}{\partial T} M_T {\rm d} T\nonumber\\
&=&  2\int_{T_1}^{T_\infty} (T_\infty - T)w(T) \frac{\partial}{\partial T} M_T {\rm d} T\nonumber\\
&+& 2\int_{T_\infty }^{T_2} (T_\infty - T)w(T) \frac{\partial}{\partial T} M_T {\rm d} T\ .
\end{eqnarray}
Integrating by parts, 
\begin{multline}
\int_{T_1}^{T_\infty} (T_\infty - T)w(T) \frac{\partial}{\partial T} M_T {\rm d} T = \\
-(T_\infty - T_1)w(T_1)M_{T_1} - \int_{T_1}^{T_\infty} \left( \frac{\partial}{\partial T} [(T_\infty - T)w(T)]\right)M_T {\rm d} T
\ .\end{multline}
Making a similar integration by parts in the last integral in (\ref{sss2}), we obtain 
\begin{multline}\label{sss5}  
 G_{T_\infty} g   =  - 2(T_\infty -T_1)w(T_1)M_{T_1}  - 2(T_2 - T_\infty )w(T_2)M_{T_2}  + \\
 2\int_{T_1}^{T_2}  \left( \frac{\partial}{\partial T} [(T_\infty - T)w(T)]\right)M_T {\rm d} T\ .
\end{multline}
Then  (\ref{sss1}) implies that
\begin{multline}\label{sss6}
\left(\eta_1 - 2(T_\infty -T_1)w(T_1)\right)M_{T_1} + 
\left(\eta_2 - 2(T_2 - T_\infty )w(T_2)\right)M_{T_2}  + \\
 \int_{T_1}^{T_2}  \left( 2\frac{\partial}{\partial T} [(T_\infty - T)w(T)] - \eta w(T) \right)M_T {\rm d} T  = 0 \ .
\end{multline}
For $c_1,c_2>0$ to be determined, suppose
\begin{equation}\label{eq4fsst}
w(T)  = \begin{cases}  c_1(T_\infty - T)^{\eta/2-1} & T\in [T_1,T_\infty)\\
  c_2 (T-T_\infty)^{\eta/2-1} & T\in (T_\infty,T_2]\end{cases}\ .      
\end{equation}
Then
$$ 2\frac{\partial}{\partial T} [(T_\infty - T)w(T)] - \eta w(T) =0$$
everywhere in $[T_1,T_\infty)\cup  (T_\infty,T_2]$, and so defining $w(T_\infty) =0$, for example, 
(\ref{sss6}) reduces to
$$
\left(\eta_1 - 2c_1(T_\infty -T_1)^{\eta/2}\right)M_{T_1} + 
\left(\eta_2 - 2c_2(T_2 - T_\infty )^{\eta/2}\right)M_{T_2}   =0\ .$$
This is satisfied for all $v$ if and only if
$$c_1 = \frac{\eta_1}{2(T_\infty -T_1)^{\eta/2}} \qquad{\rm and} \qquad 
c_2 = \frac{\eta_2}{2(T_2 -T_\infty )^{\eta/2}} \ ,$$
which yields (\ref{eq4fss}). 
The uniqueness is an immediate consequence of the fact that our equations is the forwards equation of an ergodic process; details are given in the proof of Theorem~\ref{approach} below which gives an even stronger uniqueness result. 
\end{proof}
Notice that when $w(T)$ is given by (\ref{eq4fss}), then  indeed 
$$\int_{T_1}^{T_2}w(T){\rm d}T = 1$$
for all $\eta>0$.  If we vary $\eta>0$ but keep
$$p_1 = \frac{\eta_1}{\eta} \qquad{\rm and}\qquad  p_2 = \frac{\eta_2}{\eta}$$
fixed, then  it is easy to see that
$$\lim_{\eta\to 0}w(T){\rm d}T = \delta_{T_\infty }\qquad{\rm and}\qquad  
\lim_{\eta\to \infty}w(T){\rm d}T =p_1 \delta_{T_1} + p_2\delta_{T_2}\ .$$

An entirely analogous analysis applies in the case $k > 2$, except that one must then take into account
the possibility that $T_j = T_\infty$ for some $j$. This simply introduces a point mass at $T_\infty$ into $\nu$.

\section{Approach to the NESS}\label{sect4}

We continue our study of (\ref{eq30}) and (\ref{tbgk22}), and prove Theorem~\ref{approach}.

The starting point for the  proof of Theorem~\ref{approach} is based on the elementary  fact, 
proved  in Lemma~\ref{templem2}, that
\begin{equation}\label{tfor}
T_f(t) = T_\infty + e^{-t\eta}(T_{f}(0) - T_\infty)\ ,
\end{equation} 
and that equation (\ref{eq30}) is the Kolmogorov forward equation of a stochastic process with a time dependent generator,
 but one in which the time dependence damps out exponentially fast due to (\ref{tfor}).  
 Exploiting this and making use of  a 
 variant Doeblin's Theorem to control the rate at which
memory of the initial data is lost is the basis of the proof that follows. The basic stategy is the one introduced in \cite{BL}, though here we use quantitative estimates in place of compactness arguments.
We only give the details in the case of kinetic 
Fokker-Planck equation (\ref{eq30}); the other case is similar but simpler. 

We preface the proof itself with a few further remarks on the strategy. On account of (\ref{tfor}), (\ref{eq30}) can be rewritten as
\begin{multline}\label{eq3X}
\frac{\partial}{\partial t} f(x,v,t)    +{\rm div}_x(vf(x,v,t) = G_{T_\infty + e^{-t\eta}(T_{f}(0) - T_\infty)}[f](x,v,t)\\  + \sum_{j=1}^k L_j f(x,v,t)\ .      
\end{multline}
For $t>s$, let $\widetilde f(x,v,t)$ be the solution of this equation started at time $s$ with the data $\widetilde f(x,v,s) = f_\infty(v)$. 
Since $T_\infty + e^{-t\eta}(T_{f}(0) - T_\infty) \neq T_\infty$ (except in the trivial case $T_f(0) = T_\infty$), $\widetilde f(x,v,t)$
is not independent of $t$. However,  an argument using Duhamel's formula, (\ref{tfor}) and the regularity of $f_\infty$ that is provided by 
Theorem~\ref{ssf} shows that the $L^1(\Lambda\times\R^d)$ distance
between $\widetilde f(x,v,t)$ and $f_\infty(v)$  is bounded by a fixed multiple of  $e^{-s}$ for all $t > s$. 
Since a variant of Doeblin's Theorem may be applied to show that memory of initial data is lost at an exponential rate, for
$t$ much larger than $s$, there will only be a small difference between $f(x,v,t)$ and $\widetilde f(x,v,t)$, and hence only a small difference between $f(x,v,t)$ and $f_\infty(v)$. Choosing $s=t/2$ for large $t$ then gives us the bound we seek. 

We break the proof into several lemmas, after fixing notation. 
 First, pick some large value $t_0$, to be specified later, but for now, take large to mean that $T_f(t_0)$ is very close to $T_\infty$. Consider the  stochastic process governed by (\ref{eq3X})
and started at time $s \geq t_0$ at the phase-space point $(x_0,v_0)\in \Lambda \times \R^d$ with probability $1$.  (A more detailed description of the process is provided below.) Let 
$\mathbb{P}_{s, (x_0,v_0)}$ be the law; i.e., the path-space measure of the stochastic process.  For $t > s$, and measurable
$A \subset \Lambda\times \R^d$, we are interested in the transition probabilities
\begin{equation}\label{transit}
P_{s,t}((x_0,v_0),A)  = \mathbb{P}_{s, (x_0,v_0)}(\{(x_t,v_t)\in A\})
\end{equation}
as a function of $(x_0,v_0)$, and aim to prove that the memory of  $(x_0,v_0)$ is lost at an exponential rate.  
The stochastic process consists of a Poisson stream of interactions with the thermal reservoirs, and an independent 
degenerate diffusion process between these interactions.  At the moment just after the first jump occurs,  the velocity is in a Maxwellian distribution corresponding to one of the reservoirs. Untill the next jump occurs, the particle excutes a degenerate diffusion whose transition probabilities may be explicitly estimated. If we condition on exacrtly one jump occuring in the first part of a fixed time interval, we obtain lower bounds of the particle to be at any point of the phase space. These lower bounds may then be exploints in a quantitative coupling argument based on Doeblin's Theorem. 
The next lemma  is the first step in the conditioning argument conditioning argument.

\begin{lm}\label{event}
Consider the event $E$  that in this stochastic process
there is exactly one interaction with the reservoirs in the time interval $(t_0,t_0+1/\eta]$, and none in the interval
$(t_0+1/\eta,t_0+2/\eta]$.  Then, independent of the initial data $(x_0,v_0) \in \Lambda\times \R^d$ at time $t_0$
$$\mathbb{P}_{t_0, (x_0,v_0)}(E) = e^{-2}\ .$$
\end{lm}

\begin{proof}
Since the interaction with the reservoirs occur in a 
Poisson stream with rate $\eta$, the probability of the event $E$  is
${\displaystyle \int_{t_0}^{t_0+1/\eta} \eta e^{-s\eta} e^{(s-2)\eta}{\rm d}s = e^{-2}}$.
\end{proof} 

When any interaction with the reservoirs takes place, the  velocity before the interaction  is replaced with a new velocity
chosen according the the distribution $\eta^{-1}[\eta_1M_1(v) + \eta_2M_2(v)]$, independent of what the velocity was before. 

Between interactions with the reservoirs, the motion of the particle is governed by the stochastic differential equation
\begin{eqnarray}\label{SDE1}
{\rm d} v_t &=& -v{\rm d}t  + \sqrt{2T(t)} {\rm d}w_t\nonumber \\
{\rm d}x_t &=& \phantom{-}v_t{\rm d}t 
\end{eqnarray}
where $w(t)$ is a standard Wiener process, and $T(t) = T_\infty + e^{-t\eta}(T_{f}(0) - T_\infty)$. 

We compute the distribution of $(x,v)$ at time $t_0+2/\eta$ conditional on the even $E$ defined above. 
Let $b$ denote the new velocity after the interaction with the reservoir at time $s$, and let $a$ denote the position $x_s$.
Then the distribution we seek is the distribution of $(x_{t_0+2/\eta},v_{t_0+2/\eta})$ for the solution of (\ref{SDE1}) started at
$(x_s,v_s) = (a,b)$ with probability $1$, averaged in $a$ over the distribution of $x_s$ and averaged in $b$ over the distribution
$\eta^{-1}[\eta_1M_1(v) + \eta_2M_2(v)]$.

We wish to determine the dependence of 
\begin{equation}\label{object}
{\mathbb{P}_{t_0, (x_0,v_0)}(\{(x_t,v_t)\in A\} | E ) }
\end{equation}
on $(x_0,v_0)$.  This has the following structure. Let $\mu_{(x_0,v_0),t_0,s}$, be the  
probability measure on $\Lambda$ given by
$$
\mu_{(x_0,v_0),t_0,s}(B) = \mathbb{P}_{t_0, (x_0,v_0)}(\{x_s\in B\}|E )\ ,$$
 be the law of $x_s$ at the time of
the single collision in the  interval $(t_0, t_0+1/\eta]$. Define $u = t_0+2/\eta -s$, and let $p_{b,u}(v)$ be the density for $v_{s+u}$ where $v_t$ 
is given by  (\ref{SDE1}) with $v_s = b$ with probability $1$. An explicit expression will be obtained below,
 but all that matters for us at present is that this is independent of $a$, $x_0$ and $v_0$.

 For fixed 
 $a$ and $b$  let $p_{a,b,s,s+u}(x,v)$ be the probability density of $(x_{s+u},v_{s+u})$ for the solution of (\ref{SDE1})
started at $(x_s,v_s) = (a,b)$ with probability $1$.  We shall derive 
an explicit formula for this below. Our focus will be on the conditional probability density for $x_{x+u}$ given $v_{s+u}$. 
$$p_{a,b,s,s+u}(x|v)  =\frac{p _{a,b,s,s+u}(x,v)}{p_{b,u}(v)}\ .$$

Now define the probability measure $\nu$ on $\R^d$ by
$${\rm d}\nu(b)  =  \frac{1}{\eta}[\eta_1M_1(b) + \eta_2M_2(b)]{\rm d}b\ ,$$
which is the law of the velocity immediately following an interaction with the reservoirs. Putting the components together,
the probability in (\ref{object}) is obtained by integrating
\begin{multline}\label{rep}
h_{(x_0,v_0)}(x,v) :=\\e^2\eta \int_{t_0}^{t_0+1/\eta} \int_{\R^d}\left[ \int_\Lambda[p_{a,b,s,s+u}(x| v) ] \mu_{(x_0,v_0),t_0,s}({\rm d}a ) \right] 
p_{b,u}(v) {\rm d}\nu(b) {\rm d}s
\end{multline}
over the set $A$. 
We shall obtain control   over the $(x_0,v_0)$ dependence in (\ref{object}) by quantifying the rate of coupling as follows: Consider
two copies of the process, started from $(x_0,v_0)$ and $(x_1,v_1)$ respectively. After the first interaction with the reservoir, the velocity variable jumps to a new velocity chosen independent of the starting point, and thus we obtain perfect coupling of the velocities at this time. It remains to consider the spatial coupling. To do this, we get bounds on the conditional spatial density, conditioning on the new velocity after the collision and the time of the collision. We aim to show that these conditional spatial  densities all dominate a fixed multiple of one another. From this we get a fixed minimum amount of cancelation (the effect of coupling) when subtracting transition probabilities for any two phase space starting points.  The following lemma allows us to combine velocity coupling and spatial coupling to get phase space coupling. 

\begin{lm}\label{cond}  Let $(X\times Y \times Z, \mathcal{F}_X\otimes \mathcal{F}_Y\otimes \mathcal{F}_Z)$ be a product of three measure
spaces. Let $\rho_1$ and $\rho_2$ be two probability measures on this measure space. Suppose that their marginal distributions 
on $Y\times Z$ have the following form:  There is a fixed probability measure $\nu$ on  $(Z,\mathcal{F}_Z)$ and two probability measures $\mu_1,\mu_2$ on $(Y,\mathcal{F}_Y)$ such that for all 
$B\in \mathcal{F}_Y\otimes \mathcal{F}_Z$,
$$\rho_j(X \times B) = \mu_j\otimes  \nu(B)\ .$$

Suppose also that $\rho_1$ and $\rho_2$ possess proper conditional probabilities for $Y,Z$, so that there is a representation
$$\rho_j(C) = \int_Z \rho_j(C|y,z){\rm d}\mu_j(y)\otimes {\rm d}\nu(z)\ $$
valid for all $C \in \mathcal{F}_X \otimes \mathcal{F}_Y \otimes \mathcal{F}_Z$. 
Suppose finally that there exists a constant $0 < c < \infty$ such that for all $y,y'\in Y$, all $z\in Z$, and all $A\in \mathcal{F}_X$
\begin{equation}\label{object3}
c \rho_1(A|y,z) \leq \rho_2(A|y',z) \leq \frac{1}{c} \rho_1(A|y,z) \ .
\end{equation}
Then
\begin{equation}\label{object2}
\sup\{ |\rho_1(C) - \rho_2(C)| \ :\ C \in \mathcal{F}_X \otimes \mathcal{F}_Z \ \} \leq 1 -c\ .
\end{equation}
\end{lm} 

We shall apply this in (\ref{rep}) for any two different $(x_0,v_0)$ and $(x_1,v_1)$ in $\Lambda\times \R^d$ as follows:  We take
$Z = \R^d\times \R^d$, and let the $z$ variable be $(b,v)$. 

\begin{proof} Pick any $C \in \mathcal{F}_X \otimes  \mathcal{F}_Z$. 
Then
\begin{align*}
\rho_1(C) &= \int_{Y\times Z} \left[ \int_X 1_C(x,y,z)\rho_1({\rm d}x|y,z) \right]{\rm d}\mu_1 (y){\rm d}\nu (z) \\
&= \int_Y \left( \int_{Y\times Z} \left[ \int_X 1_C(x,y,z)\rho_1({\rm d}x|y,z) \right]{\rm d}\mu_1 (y){\rm d}\nu (z)\right) 
{\rm d}\mu_2 (y')\\
&= \int_Y \left( \int_{Y\times Z} \left[ \int_X 1_C(x,y,z)\rho_1({\rm d}x|y,z) \right]{\rm d}\mu_2 (y'){\rm d}\nu (z)\right) 
{\rm d}\mu_1 (y)\\
&\geq c\int_Y \left( \int_{Y\times Z} \left[ \int_X 1_C(x,y',z)\rho_1({\rm d}x|y',z) \right]{\rm d}\mu_2 (y'){\rm d}\nu (z)\right) 
{\rm d}\mu_1 (y)\\
&= c\rho_2(C)\ ,
\end{align*}
where the first equality is trivial, the second is valid by the Fubini-Tonelli Theorem, and the inequality is (\ref{object3}), together with the fact that
$1_C$ does not actually depend on $y$. Then (\ref{object2}) follows directly. 
\end{proof} 

In the context of (\ref{rep}), we apply Lemma~\ref{cond} for {\em fixed $s\in [t_0,t_0+1/\eta]$} as follows: We take $Z = \R^d\times \R^d$ with variables $b$ and $v$. We take $Y = \Lambda$ with variable $a$, and we take $X =\Lambda$ with variable $x$. We choose $C\in \mathcal{F}_X\otimes \mathcal{F}_Z$, and in fact, independent of the $b$ component of $Z$. In the next subsection we shall derive an explicit formula for the conditional probability $p_{a,b,s,s+u}(x| v)$, and along the way,  
$p_{b,u}(v)$. Using this formula, we shall show that (\ref{object3}) is satisfied.  The constant $c$ in 
(\ref{object3}) will be shown to depend only on $T_\infty$ and $\eta$.  Then, integrating in $s$, we shall have proved:

\begin{lm}\label{uniform} For fixed $a$ and $b$, and for fixed $u> 0$, let $p_{a,b,s,s+u}(x,v)$  be the probability 
density of $(x_{s+u},v_{s+u})$ for the solution of (\ref{SDE1})
started at $(x_s,v_s) = (a,b)$ with probability $1$. Let $p_{a,b,s,s+u}(x|v)$ be the conditional density of $x_{s+u}$
given $v_{s+u} = v$.  Then there is an explicitly computable constant 
$C_{\eta,T_\infty,|\Lambda|}>0$ depending only on $\eta$, $T_\infty$ and $|\Lambda|$
and an explicitly computable $s_0 < \infty$ such that for all $s\geq s_0$ and all $u\in [1/\eta,2/\eta]$,
the following  bound holds uniformly in $x\in \Lambda$:
\begin{equation}\label{joint26}
C_{\eta,T_\infty,|\Lambda|}  \leq p_{a,b,s,s+u}(x| v) \leq \frac{1}{C_{\eta,T_\infty,|\Lambda|}}\ .
\end{equation}
\end{lm}

The proof is given at the end of the next subsection. 

\subsection{Estimates for the degenerate diffusion}

We now derive an explicit formula for this probability, which 
 will be an analog of  a well-known formula obtained by Kolmogorov \cite{K34}  in 1934. The variant we need, taking into account our time-varying temperature, is easily derived using stochastic calculus as done by McKean in 
\cite{McK}. 
 We will write down the formula in the case $d=1$ in order to keep the notation simple. It will be clear from these formulas  and their derivation that the conclusions we draw from them are valid in all dimensions.

If $x(s) = a$ and $v(s) = b$, we have that for $t>s$
\begin{equation}\label{SDE1a}
v_t = e^{s-t}b + \int_s^t e^{r-t}\sqrt{2T(r)}{\rm d}w_r\ .
\end{equation}  

In our application,  $\Lambda$ is a circle of circumference $L$. However, we may solve (\ref{SDE1}) for $x_t\in \R$, and then later ``wrap'' this into the circle. So for the moment, let us take $x_t\in \R$. 
We  may then integrate to find $x_t$. To do so, note that 
$$ \int_s^t\int_s^u e^{r-u}\sqrt{2T(r)}{\rm d}w_r{\rm d}u = \int_s^t (1- e^{r-t})\sqrt{2T(r)}{\rm d}w_r \ .$$
Therefore, 
\begin{equation}\label{SDE1b}
x_t =  a + (1-e^{s-t})b + \int_s^t (1- e^{r-t})\sqrt{2T(r)}{\rm d}w_r \ .
\end{equation} 

The random variables $x_t$ and $v_t$ are evidently Gaussian, and their joint distribution is determined by their means, variances, and correlation. 
Let $\mu_x$ and $\mu_v$ denote the means of $x_t$ and $v_t$ respectively. Let $\sigma_x$ and $\sigma_v$  denote the standard deviations of $x_t$ and $v_t$ respectively. Finally, let $\rho$ denote their correlation, which means that
$$\rho \sigma_x\sigma_v = {\mathbb{E}}(x_t - \mu_x)(v_t - \mu_v)\ .$$
The joint density of $(x_t,v_y)$, again taking $x_t$ to be $\R$ valued, is the probability density given by 
\begin{multline}\label{joint}
f(x,v) = \frac{1}{2\pi \sigma_x\sigma_v \sqrt{1-\rho^2}}\times 
\\ \exp\left(\frac{-1}{2(1-\rho^2)}\left[ \frac{(x-\mu_x)^2}{\sigma_x^2} +
 \frac{(v-\mu_v)^2}{\sigma_v^2}  - \frac{2\rho (x-\mu_x)(v - \mu_v)}{\sigma_x\sigma_v}\right]\right)\ .
 \end{multline}
 Completing the square, we obtain the alternate form
 \begin{multline}\label{joint2}
f(x,v) = \frac{1}{2\pi \sigma_x\sigma_v \sqrt{1-\rho^2}}\times \\
 \exp\left(\frac{-1}{2(1-\rho^2)}\left[ \left( \frac{(x-\mu_x)}{\sigma_x} - \rho \frac{(v-\mu_v)}{\sigma_v}\right)^2 +
 (1-\rho^2)\frac{(v-\mu_v)^2}{\sigma_v^2}  \right]\right)\ .
 \end{multline}
Evidently, the conditional density of $x_t$ given $v_t = v$ is
\begin{equation}\label{joint3}
f(x|v) := \frac{1}{ \sigma_x \sqrt{2\pi(1-\rho^2)}}
 \exp\left(\frac{-1}{2(1-\rho^2)}\left( \frac{(x-\mu_x)}{\sigma_x} - \rho \frac{(v-\mu_v)}{\sigma_v}\right)^2 
  \right)\ ,
 \end{equation}
 and the density of $v_t$ is
 \begin{equation}\label{joint4}
 f(v) = \frac{1}{\sigma_v \sqrt{2\pi}} \exp\left(-\frac{(v-\mu_v)^2}{2\sigma_v^2} \right)\ .
  \end{equation}
 Note that $f(x,v)= f(x|v)f(v)$.   Now we wrap this density onto $\Lambda$ which we  identify with the interval $(-L/2,L/2] \subset \R$.
 The wrapped density is, for $x\in (-L/2,L/2]$,
  \begin{equation}\label{joint5} 
  \widetilde f(x,v) = \sum_{k\in \Z} f(x+kL,v)f(v)\ .
    \end{equation}
(It is evident that for fixed $v$, the lower bound on $\widetilde f(x,v)$ as a function of $x$ will go to zero as $L$ goes to infinity.)
Recall that for fixed $a$ and $b$, and for fixed $u> 0$, we have defined $p_{a,b,s,s+u}(x,v)$ to be the probability 
density of $(x_{s+u},v_{s+u})$ for the solution of (\ref{SDE1})
started at $(x_s,v_s) = (a,b)$ with probability $1$.  By the computations above, $p_{a,b,s,s+u}(x,v)$ is obtained
by substituting $\sigma_x(s+u)$, $\sigma_v(s+u)$ and $\rho(s+u)$ into (\ref{joint5}).   We now compute these quantities using the Ito isometry for the variance and covariance.

 From (\ref{SDE1a}) and (\ref{SDE1b}) we readily compute
 \begin{equation}\label{joint40}
 \mu_x = a + (1- e^{s-t}b) \quad{\rm and}\quad  \mu_v = e^{s-t}b\ . 
 \end{equation}
 Next, we compute:
 \begin{align}\label{joint41}
 \sigma_v^2(s+u) &= \int_s^{s+u} e^{2(r-(s+u))}2T(r){\rm d}r\nonumber\\
 &= (1- e^{-2u})T_\infty + \frac{2(T_{f(0)} - T_\infty)}{2-\eta}e^{-\eta s} (e^{-\eta u} - e^{-2u})\ .
 \end{align}
 (Note that a limit must be taken at $\eta =2$.)
 
 \begin{align}\label{joint42}
 \sigma_x^2(s+u) &= \int_s^{s+u} (1- e^{r-(s+u)})^2 2T(r){\rm d}r\\
 &= 4(e^{-u} -1 +u)T_\infty - (e^{-2u} -1 +2u)T_\infty\nonumber\\
 &+  \frac{2(T_{f(0)} - T_\infty)}{\eta(\eta-1)(\eta-2)}\times \nonumber\\
 &\phantom{+}e^{-\eta s} ( \eta^2 (1 - e^{-u})^2  -\eta  
 (1-e^{-u})(3-e^{-u}) + 2(1- e^{-\eta u}))\ .
 \end{align}
(Note that a limit must be taken at $\eta =1$ or $\eta =2$. Note also that $\sigma_x^2(s+u)$ is of order $u^3$ for small $u$, 
as one would expect from Kolmogorov's formula, and of order $u$ for large $u$.)  Next, we compute the covariance.
For $t = s+u$, 
 
\begin{align} \label{joint43} {\mathbb{E}}(x_t - \mu_x)(v_t - \mu_v) &= 
\int_s^t e^{r-t}(1- e^{r-t}) 2T(r){\rm d}r\nonumber\\
&= T_\infty (1-e^{-u} )^2 \nonumber\\
&+  \frac{2(T_{f(0)} - T_\infty)}{(\eta-1)(\eta-2)}e^{-\eta s}( (\eta-2)  (1 - e^{-u})  -(\eta -1) e^{-2u} + e^{-\eta u})\ .
\end{align}

In our application, we are concerned with $s\in (t_0 + 1/\eta]$ and $t = t_0 + 2/\eta$, and hence with $u\in [1/\eta,2/\eta]$.
For such $u$, and large $t_0$ (and hence large $s$, the following approximations are accurate up to exponentially small
(in $t_0$) percentage-wise corrections:

\begin{align}
 \sigma_v^2(s+u) 
 &\approx (1- e^{-2u})T_\infty \label{joint23a} \\
 \sigma_x^2(s+u) &\approx  [4(e^{-u} -1 +u) - (e^{-2u} -1 +2u)]T_\infty\label{joint23b}\\
 \rho(s+u) &\approx  \frac{ (1-e^{-u} )^2 }
 {\sqrt{ (1- e^{-2u}) (4(e^{-u} -1 +u)- (e^{-2u} -1 +2u))} }\label{joint23c}\ .
\end{align}
Define $\hat  \rho(s+u)$ to be the quantity on the right side of (\ref{joint23c}). One readily checks that 
$$\hat  \rho(s+u) = \frac{\sqrt{3}}{2} + {\cal{O}}(u)$$
for small $u$, that $ \hat  \rho(s+u) \sqrt{1+ u/5}$ is monotone decreasing, and that 
$\lim_{u\to\infty} \hat \rho(s+u) \sqrt{1+ u/5} =1/ \sqrt{10}$.  Altogether,
\begin{equation}\label{joint240}
\frac{\sqrt{3}}{2}\frac{1}{\sqrt{5+u}} \geq \hat \rho(s+u) \geq  \frac{1}{\sqrt{2}}\frac{1}{\sqrt{5+u}}\ .
\end{equation}
It now follows that for all $s$ sufficiently large, and all $u\in [1/\eta,2/\eta]$,
\begin{equation}\label{joint24}
\frac{\sqrt{3}}{2}\frac{1}{\sqrt{4+1/\eta}} \geq \rho(s+u) \geq  \frac{1}{\sqrt{2}}\frac{1}{\sqrt{6+2\eta}}\ .
\end{equation}

Likewise, define $\hat \sigma_x^2(s+u)$ to be the right side of (\ref{joint23b}). Simple calculation show that 
$\hat \sigma_x^2(s+u)(1+2u)^2/u^3$ is monotone for $u>0$ increasing with
$$ 8T_\infty  \geq   \hat \sigma_x^2(s+u)(1+2u)^2/u^3  \geq  \frac{2}{3}T_\infty \ .$$
It now follows that for all $s$ sufficiently large, and all $u\in [1/\eta,2/\eta]$,
\begin{equation}\label{joint25}
64\frac{\eta^{-3}}{(1+2/\eta)^2}T_\infty  \geq \sigma^2_x(s+u) \geq  \frac{2}{\sqrt{3}}\frac{\eta^{-3}}{(1+ 5/\eta)^2}T_\infty\ .
\end{equation}

\begin{proof}[Proof of Lemma~\ref{uniform}] By the computations above, for any $x\in [-L/2,L/2]$ by substituting into 
$$\sum_{k\in \Z} \frac{1}{ \sigma_x \sqrt{2\pi(1-\rho^2)}}
 \exp\left(\frac{-1}{2(1-\rho^2)}\left( \frac{(x+ kL -\mu_x)}{\sigma_x} - \rho \frac{(v-\mu_v)}{\sigma_v}\right)^2
  \right)
  $$ the appropriate values of $\mu_x$, $\mu_v$, $\sigma_x$, $\sigma_v$ and $\rho$, given by (\ref{joint40}), 
  (\ref{joint41}), (\ref{joint42}) and (\ref{joint43}).
  
  Now, whatever, the value of $v$, there is some $k\in Z$ so that 
  $$|x+ kL -\mu_x - (\sigma_x/\sigma_v)\rho (v-\mu_v)| \leq L/2\ .$$
  Retaining only this term in the sum, 
  $$p_{a,b,s,s+u}(x|v) \geq  \frac{1}{ \sigma_x \sqrt{2\pi(1-\rho^2)}}
  \exp\left(\frac{-L^2}{8(1-\rho^2)\sigma_x^2} \right)\ .$$ 
Now using the upper and lower bounds for $\rho$ and $\sigma_x$ that are given in 
(\ref{joint24}) and (\ref{joint25}), which are valid for all $s\geq s_0$ that may be explicitly computed keeping track of
constants leading up to (\ref{joint24}) and (\ref{joint25}).  The corresponding uniform upper bound is readily derived by estimating the sum in $k$, which converges extremely rapidly. 
\end{proof}

It is evident from the proof that the analogous lemma for the $d$ dimensional version of our process is also valid.

\subsection{Bounds on transition functions}

Recall that for any $t> s \geq t_0>0$, and any measurable $A \subset \Lambda \times \R^d$, and any $(x_0,v_0)\in \Lambda \times \R^d$,
$P_{t_0,t}((x_0,v_0),A)$ is the probability that our original stochastic process (with interactions with the reservoirs), started at $(x_0,v_0)$ at time $s$ satisfies
$(x_t,v_t)\in A$. Let 
$\mathbb{P}_{s, (x_0,v_0)}$ be the path-space measure of the stochastic process. Then, with $E$ being the event considered
in Lemma~\ref{event}, 
\begin{align}\label{cond1}
P_{t_0,t}((x_0,v_0),A)  &= \mathbb{P}_{t_0, (x_0,v_0)}(\{(x_t,v_t)\in A\})\\
&= \mathbb{P}_{t_0, (x_0,v_0)}(\{(x_t,v_t)\in A\}\cap E ) + \mathbb{P}_{t_0, (x_0,v_0)}(\{(x_t,v_t)\in A\}\cap E^c)\nonumber
\end{align}
By Lemma~\ref{event}, we can express $\mathbb{P}_{t_0, (x_0,v_0)}(\{(x_t,v_t)\in A\}\cap E ) $ in terms of the conditional probabilities we have estimated in the previous subsections:
$$ \mathbb{P}_{t_0, (x_0,v_0)}(\{(x_t,v_t)\in A\}\cap E )  = e^{-2} \mathbb{P}_{t_0, (x_0,v_0)}(\{(x_t,v_t)\in A \ |\  E ) \ .$$

\begin{thm}\label{deficit}  There is an explicitly computable constant 
$C_{\eta,T_\infty,|\Lambda|}>0$ depending only on $\eta$ and $T_\infty$
and an explicitly computable $t_0 < \infty$ such that for $t =t_0+2/\eta$,  and all
$(x_0,v_0)$  and $(x_1,v_1)$ in $\Lambda\times \R^d$
$$
\sup\{ | \mathbb{P}_{t_0, (x_0,v_0)}(\{(x_t,v_t)\in A\}  )  - \mathbb{P}_{t_0, (x_1,v_1)}(\{(x_t,v_t)\in A\}  )| \} \leq
1 - e^{-2}
C_{\eta,T_\infty,|\Lambda|}\ ,$$
where the supremum is taken over all measurable subset of $\Lambda\times \R^d$, 
\end{thm} 

\begin{proof}  By Lemma~\ref{event},
\begin{multline}
\mathbb{P}_{t_0, (x_0,v_0)}(\{(x_t,v_t)\in A\}  =\\ e^{-2} \mathbb{P}_{t_0, (x_0,v_0)}(\{(x_t,v_t)\in A\} | E )
+ (1-e^{-2}) \mathbb{P}_{t_0, (x_0,v_0)}(\{(x_t,v_t)\in A\} | E^c )\ .
\end{multline}
By Lemma~\ref{uniform}, (\ref{object3}) is satisfied with $c = C_{\eta,T_\infty,|\Lambda|}$ when we apply Lemma~\ref{cond} to the probabilities
$\mathbb{P}_{t_0, (x_0,v_0)}(\{(x_t,v_t)\in A\}) $ given by (\ref{object}) and (\ref{rep}).
Therefore,
\begin{align*}
&|\mathbb{P}_{t_0, (x_0,v_0)}(\{(x_t,v_t)\in A\}) - \mathbb{P}_{t_0, (x_1,v_1)}(\{(x_t,v_t)\in A\})|
\\ 
& \leq e^{-2}|\mathbb{P}_{t_0, (x_0,v_0)}(\{(x_t,v_t)\in A\}|E) - 
\mathbb{P}_{t_0, (x_1,v_1)}(\{(x_t,v_t)\in A\}|E)|  \\
&  + (1-e^{-2}) |\mathbb{P}_{t_0, (x_0,v_0)}(\{(x_t,v_t)\in A\}|E^c) - 
\mathbb{P}_{t_0, (x_1,v_1)}(\{(x_t,v_t)\in A\}|E^c)|\\
& \leq e^{-2} \Big|\mathbb{P}_{t_0, (x_0,v_0)}(\{(x_t,v_t)\in A\}|E) - 
\mathbb{P}_{t_0, (x_1,v_1)}(\{(x_t,v_t)\in A\}|E)\Big|  + (1-e^{-2})\\
&\leq e^{-2}(1 - C_{\eta,T_\infty,|\Lambda|}^2) + (1- e^2) = 1 - e^{-2}C_{\eta,T_\infty,|\Lambda|}^2)\ .
\end{align*}
Theorem~\ref{deficit} now follows from what has been said in the paragraph preceding it. 
\end{proof}

We next  recall a form of Doeblin's Theorem
 in the context of temporally non-homogeneous processes.

Let $P_{s,t}(z,A)$, $t> s$, be a family of Markov kernels on a measure space $(Z,\mathcal{F})$
such that  for $r < s < t$
\begin{equation}\label{mar}
P_{r,t}(z,A) = \int_Z P_{r,s}(z,{\rm d}y) P_{s,t}(y,A)\ ,
\end{equation}
and we suppose also that for fixed $A$, $s$ and $t$, $P_{s,t}(z,A)$ is a continuous function of $z$.
Define the quantity
$$\rho_{s,t} =   \sup_{z,y\in Z}\sup_{A\in \mathcal{F}} \left\{|P_{s,t}(z,A) - P_{s,t}(y,A)|\right\}\ .$$
The following lemma is an adaptation of a proof in Varadhan's text \cite{Var}:

\begin{lm}\label{doeb} For all $r < s < t$, 
$\rho_{r,t} \leq \rho_{r,s}\rho_{s,t}$.
\end{lm}

\begin{proof} By (\ref{mar}), 
\begin{equation}\label{mar2}
|P_{r,t}(z,A) - P_{r,t}(y,A)|  =   
\left|\int_Z  P_{r,s}(z,{\rm d}w) P_{s,t}(w,A)  -   \int_Z  P_{r,s}(y,{\rm d}w) P_{s,t}(w,A)\right| 
\end{equation}
To write this more compactly, introduce the continuous function $f(w) = P_{s,t}(w,A)$, and 
let $\nu$ denote the signed measure
$${\rm d}\nu(w)  = P_{r,s}(z,{\rm d}w) - P_{r,s}(y,{\rm d}w)\ .$$   
Then (\ref{mar2}) becomes
\begin{equation}\label{mar3}
|P_{r,t}(z,A) - P_{r,t}(y,A)|  = 
\left|\int_Z  f(w) {\rm d}\nu(w) \right| \ .
\end{equation}
Define $\|\nu\| = \sup_{A\in \mathcal{F}}|\nu(A)|$. Then if $\nu = \nu_+ - \nu_-$ is the Hahn decomposition of $\nu$ into its positive and negative parts,  $\nu_+(Z) = \nu_-(Z)$ (since $\nu(Z) = 0$), and
$$\|\nu\| = \nu_+(Z) \ .$$
Let ${\rm d}|\nu|$ be the measure defined by  ${\rm d}|\nu| = {\rm d}\nu_+ + {\rm d}\nu_-$. Then
\begin{equation}\label{mar10}
|\nu|(X) = 2\|\nu\| \leq 2 \rho_{t,s}\ .
\end{equation}
If $\varphi$ is any continuous function on $Z$, we have
\begin{equation}\label{mar6}
\left|\int_Z \varphi {\rm d}\nu\right| \leq   \int_Z |\varphi|  {\rm d}|\nu| \leq \|\varphi\|_\infty |\nu|(Z) = 2
\|\varphi\|_\infty\|\nu\|\ .
\end{equation}
By hypothesis
$$|f(z_1) - f(z_2)| \leq \rho_{s,r} \quad{\rm for\ all}\quad z_1,z_2\in Z\ . $$
It follows that the range of $f$ is contained in an interval of width at most $\rho_{s,r}$, and hence
\begin{equation}\label{mar11}
\inf_{c\in \R}\|f - c\|_\infty \leq \frac12 \rho_{s,r}\ .
\end{equation}
Since $P_{s,t}(z,{\rm d}w) $ and $P_{s,t}(y,{\rm d}w) $ are both probability measures, for any constant $c\in \R$,
\begin{equation}\label{mar7}
\int_Z f(w){\rm d}\nu(w) =  \int_Z (f(w)-c){\rm d}\nu(w) \ ,
\end{equation}
and hence (\ref{mar3}) becomes
\begin{eqnarray}\label{mar4}
|P_{r,t}(z,A) - P_{r,t}(y,A)| &=&
\inf_{c\in \R} \left|\int_X  (f(w) - c) {\rm d}\nu(w) \right| \nonumber\\
&\leq& \inf_{c\in \R} \|f-c\|_\infty |\nu(Z)|\nonumber.
\end{eqnarray}
Now using the estimate (\ref{mar10}) and (\ref{mar11}) we obtain the result. 
\end{proof}

The following lemma is now a direct consequence of Theorem~\ref{deficit} and the definitions made just above. 

\begin{lm}\label{BL}  Let $P_{s,t}((x,v),A)$ be the transition function for the process associated to (\ref{eq3X}). There exist 
explicitly computable 
$\delta>0$,  $t_0>0$ and $t_1>0$ such that for all $A$ and all
$t>s \geq t_0$ with $t-s \geq t_1$, 
\begin{equation}\label{mar20}
\rho_{s,t} \leq 1-\delta\ .
\end{equation}
\end{lm}

We are now ready to prove Theorem~\ref{approach}.

\begin{proof}[Proof of Theorem~\ref{approach} for the kinetic Fokker-Planck equation]
Combining Lemma~\ref{doeb} (For $Z = \Lambda \times \R^d$) and Lemma~\ref{BL}, whenever $t-t_0 \geq  nt_1$, 
$$\rho_{t_0,t} \leq (1-\delta)^n\ .$$

Let $\mu^{(1)}$ and $\mu^{(2)}$ are any two probability measures on $\Lambda \times \R^d$, and  use them to initialize the 
 the Markov process associated to our
family of transition kernels at time $t_0$.  For $j=1,2$
define
$$\mu^{(j)} _t(A)   = \int_{\Lambda \times \R^d}{\rm d}\mu_j(z) P_{t_0,t}(z,A)\ .$$
Then for all $t> t_0$
$$\|\mu^{(1)}_t  - \mu^{(2)}_t\| \leq  \rho_{t_0,t}\ .$$
Thus, when (\ref{mar20}) is valid, the memory of the initial condition is washed out exponentially fast.  Of course this by itself does not imply convergence to a stationary state, and if the time inhomogeneity of our process had a strong oscillatory character, for example, we would not expect convergence to a stationary state. However, our process has an asymptotic temporal homogeneity property; namely the generator converges exponentially fast to the generator of a stationary process with an invariant measure $\mu_\star$. We use this to show that for any $\mu_0$, $\mu_t$ converges exponentially  fast to $\mu_\star$. 

 Let $R_{s,t}(z,A)$ be a homogenous family of 
Markov kernels on $(Z,\mathcal{F})$, where in our case $Z = \Lambda \times \R^d$. That is, for $s < t < u$
$$R_{s,u}(z,A) = \int_Z R_{s,t}(z,{\rm d}z')R_{t,u}(z',A)\ .$$
In our application, $R_{s,t}(z,A)$ will be the transition kernel for the Markov process corresponding to  (\ref{eq4}), 
however, it is convenient  to continue  the discussion in a general context for now. 

Let $L$ be the generator for the process governed by $R_{s,t}(z,A)$.  Then for $s < t$, the Kolmogorov backward equation
\begin{equation}\label{koleq1}
-\frac{\partial}{\partial s}R_{s,t}(z,A)  =  L R_{s,t}(z,A)\ ,
\end{equation}
with the final condition $R_{t,t}(z,A) =1_A(z)$, where $1_A$ is the indicator function of $A$. 
For $t>s$, $R_{s,t}(z,\cdot)$, considered as a time dependent measure, satisfies the Kolmogorov forward equation
\begin{equation}\label{koleq2}
\frac{\partial}{\partial t}R_{s,t}(z,\cdot)  =  L ^*R_{s,t}(z,\cdot)\ ,
\end{equation}
where $L^*$ is the adjoint of $L$, meaning that for any bounded continuous function $\varphi$,
\begin{equation}\label{koleq3}
\int L ^*R_{s,t}(z,{\rm d}z')\varphi(z') = \int R_{s,t}(z,{\rm d}z')L\varphi(z')\ .
\end{equation}

Let $K_t$ be the generator of the inhomogeneous process described by $P_{s,t}$. Note that for any
fixed measurable set $A$ in phase space,  $P_{s,t}((x,v),A)$ with $s< t$,   satisfies the Kolmogorov backward equation
\begin{equation}\label{koleq4}
-\frac{\partial}{\partial s}P_{s,t}((x,v),A)  =  K_s P_{s,t}((x,v),A)\ ,
\end{equation}

Then by the Fundamental Theorem of Calculus, (\ref{koleq2}),
(\ref{koleq3}) and (\ref{koleq4})
\begin{align*}
P_{s,t}(z,A) - R_{s,t}(z,A) &= -\int_s^t\left(\frac{{\rm d}}{{\rm d}u} \int R_{s,u}(z,z') P_{u,t}(z,,A){\rm d}z\right) {\rm d}u\\
&= -\int_s^t\left( \int R_{s,u}(z,z') (L- K_u)P_{u,t}(z,A){\rm d}z\right) {\rm d}u
\end{align*}
Let $B_t = K_t - L$, the difference of the generators.  Then we can rewrite the result of this computation as
\begin{equation}\label{koleq5}
P_{s,t}(z,A) - R_{s,t}(z,A)  = \int_s^t\left( \int R_{s,u}(z,z') B_u P_{u,t}(z,A){\rm d}z\right) {\rm d}u\ .
\end{equation}

For any measure $\mu$, let $\mu P_{s,t}$ denote the measure $\mu P_{s,t}(A) = \int {\rm d}\mu(z) P_{s,t}(z,A)$, and likewise for
$\mu R_{s,t}$. Then
if $\mu_\star$ is an invariant measure for $R$, meaning that $\mu_\star = \mu_\star R_{s,t}$, we have from (\ref{koleq5})
that
\begin{equation}\label{koleq6}
\mu_\star P_{s,t}(A) - \mu_\star(A)    = \int_s^t\left( \int \mu_\star(z) B_u P_{u,t}(z,A){\rm d}z\right) {\rm d}u\ .
\end{equation}

We now apply the general formula (\ref{koleq6}) in our specific context, in which 
$B_u = (T(u) - T_\infty)\Delta_v$ where $(T(s) - T_\infty) = Ce^{-cu}$ for some $C$ and $c>0$ and $\mu_\star = f_\infty(v){\rm d}v$.
Then
\begin{equation}\label{remainder}
\|B_u^*\mu_\star\| \leq Ce^{-cu}\|\Delta_v f_\infty \|_{L^1}\ .
\end{equation}
Crucially, $\|\Delta_v f_\infty \|_{L^1} < \infty$ by the explicit formula we have found for $f_\infty$.
Integrating, for all $t < t_2$,
$$
\left\Vert    \int_{t}^{t_2} (B_s^*\mu_\star)  P_{s,t_2}{\rm d}s \right\Vert \leq \int_{t}^{t_1} Ce^{-cu}{\rm d}u \leq \frac{C}{c} e^{-ct}\ .
$$
Consequently,
$$\|\mu_\star P_{t,t_2} -\mu_\star\| \leq   \frac Cce^{-ct}\ .$$

\if false
 Suppose one can prove that 
$$\|\mu_\star P_{t_0,t} - \mu_\star\|$$
tends to zero exponentially fast. Granted this, let $\mu_0$ be any probability measure on $Z$. Then
$$\|\mu_0  P_{t_0,t}- \mu_\star\| \leq \| \mu_\star P_{t_0,t} - \mu_0 P_{t_0,t}\| + 
\|\mu_\star P_{t_0,t} - \mu_\star\|\ .$$
When (\ref{mar20}) is valid, the first term on the right tends to zero exponentially 
fast by what we have explained above. 

It remains to prove that  $\|P_{t_0,t}\mu_\star - \mu_\star\|$tends to zero exponentially fast. We return to our specific context.

while $R_t((x,v),A)$  satisfies the backward equation
$$
-\frac{\partial}{\partial s}R_{t-s}((x,v),A)  =  L R_{t-s}((x,v),A)\ ,
$$

By Duhamel's formula, for any $t_2> t > t_0$,
$$P_{t,t_2}  = R_{t_2-t} + \int_{t}^{t_2} R_{s-t} B_s P_{s,t_0}{\rm d}s \ .$$
Therefore, applying both sides to $\mu_\star$,
$$
\mu_\star P_{t,t_2} = \mu_\star  + \int_{t}^{t_2} (B_s^*\mu_\star)  P_{s,t_2}{\rm d}s
$$
Let $f_\infty(v) (v)$ be the density of $\mu_\star$. Then
\fi

Now for any probability measure $\mu_0$ on phase space, and any  $s> t_0$ and any $t > t_0 + n(t_1-t_0)+ s$,
$$\|\mu_0 P_{0,t}  - \mu_\star \| =   \|(\mu_0 P_{0,s} )P_{s,t}  - \mu_\star P_{s,t} \| + \|\mu_\star  P_{s,t}  - \mu_\star \|$$
By Lemma~\ref{BL}, $\|(\mu_0 P_{0,s})P_{s,t}  - \mu_\star P_{s,t} \| \leq (1-\delta)^n$.  By the computations above,
$\|\mu_\star  P_{s,t}  - \mu_\star \| \leq (C/c)e^{-cs}$. 

Now for  $t> t_0$, we have that
$$\|\mu_0 P_{0,t}  - \mu_\star \|  \leq (C/c)e^{(t-t_0)/2}  + \frac{1}{1-\delta} e^{\ln(1-\delta)(t-t_0)/2}\ .$$
This gives us the exponential relaxation. 
\end{proof}

{\bf Acknowledgements} E.A.~Carlen  is partially supported by  NSF Grant Number  DMS  1501007. 
J.L.~Lebowitz is partially supported by AFOSR Grant FA9550-16-1-0037 and NSF Grant Number DMR1104501


\begin{thebibliography}{30}

\small{

\bibitem{AEP} L.~Arkeryd, R.~Esposito and M.~Pulvirenti, \textit{The Boltzmann equation for weakly 
inhomogenous initial data}. Comm. Math. Phys., {\bf 111}, 393-407, 1987 

\bibitem{BL} P.~G.~Bergmann and J.~L.~Lebowitz, \textit{New Approach to Nonequilibrium Process}. 
Physical Review, 99, 578-587, 1955. 


\bibitem{CLM} E.~A.~Carlen,  J.~L.~Lebowitz and C.~Mouhot,  \textit{Exponential approach to, 
and properties of, a non-equilibrium steady state in a dilute gas.} Braz. J. Probab. Stat., {\bf  29},  (2015), 372--386.

\bibitem{Cer} C.~Cercignani, \textit{The Boltzmann Equations and its Applications}, Berlin, Springer Verlag, 1987.

\bibitem{DV} L.~Desvillettes and C.~Villani,  \textit{On the trend to global equilibrium for spatially 
inhomogeneous kinetic systems: the Boltzmann equation}. Invent. Math. {\bf 159} ,  245-316, 2005

\bibitem{EGKM}R.~Esposito, Y.~Guo, C.~Kim and R.~Marra, \textit{Non-Isothermal Boundary in 
the Boltzmann Theory and Fourier Law.} Comm. Math. Phys. {\bf 323} (2013), 359-385

\bibitem{FM} N.~Fournier and S. M\'el\'eard \textit{A Markov Process Associated with a Boltzmann 
Equation Without Cutoff and for Non-Maxwell Molecules}, J. Stat Phys., {\bf 104}, (2001) 583-604


\bibitem{FM2} N.~Fournier and S. M\'el\'eard \textit{A stochastic particle numeric method for a 
3D Boltzmann equation without cuttoff}, Math. Computation, {\bf 71}, (2001) 583-604

\bibitem{GMM} M.~Gualdani, S.~Mischler and C.~Mouhot, \textit{Factorization for non-symmetric operators and exponential H-theorem}, arXiv preprint arXiv:1006.5523. 

\bibitem{HN} F.~H\'erau, F. ~Nier: \textit{Isotropic hypoellipticity and trend to equilibrium for the
Fokker-Planck equation with a high-degree potential}. Arch. Ration. Mech.
Anal. {\bf 171},  (2004) pp.151Ð218

\bibitem{K34}A.~N.Kolmogorov, \textit{Zuff\"allige  Bewegungen}, Ann. Math. II, {\bf 35} (1934), 116-117.

\bibitem{LB} J.~L.~Lebowitz and P~G.~Bergmann \textit{Irreversible Gibbsian Ensembles} Annals of Physics, {\bf 1}, no 1, (1957) 1-23. 


\bibitem{McK} H.~P.~McKean, \textit{A winding problem for a resonator driven by white noise}, J. Math. Kyoto Univ.
{\bf 2}, no 2, (1963) 227-235.


\bibitem{Var} 
S.~R.~S.~Varadhan: \textit{Probability Theory} New York University, New York (2001).

\bibitem{Vil} 
C.~Villani: \textit{Hypocoercivity} Memoirs of the A.M.S. {\bf 202}, no. 950, A.M.S., Providence RI, (2009).



}







\end{thebibliography}
\end{document}